\documentclass[12pt]{article}

\usepackage{amsfonts}
\usepackage{amsmath}
\usepackage{graphicx}
\usepackage{epstopdf}
\usepackage{fancyhdr}
\usepackage{courier}
\usepackage{float}
\usepackage{amsthm}
\usepackage{mathrsfs}
\usepackage{framed, bm}
\usepackage{natbib}

\usepackage[top=1.0in, bottom=1.3in, left=1in, right=1in]{geometry}

\theoremstyle{plain}
\newtheorem{theorem}{Theorem}[section]
\theoremstyle{plain}
\newtheorem{lemma}{Lemma}[section]
\theoremstyle{plain}
\newtheorem{proposition}{Proposition}[section]
\theoremstyle{plain}
\newtheorem{algorithm}{Algorithm}[section]

\theoremstyle{definition}

\theoremstyle{remark}
\newtheorem{remark}{Remark}[section]

\newcommand{\E}{\mathrm{E}}

\newcommand{\Var}{\mathrm{Var}}

\newcommand{\Cov}{\mathrm{Cov}}

\newcommand\independent{\protect\mathpalette{\protect\independenT}{\perp}}
\def\independenT#1#2{\mathrel{\rlap{$#1#2$}\mkern2mu{#1#2}}}

\title{Locally weighted Markov chain Monte Carlo}
\author{
\textsc{Espen Bernton} \\
\textit{ Department of Statistics, Harvard University }\\
\texttt{ebernton@g.harvard.edu} \and  \textsc{Shihao Yang} \\
\textit{ Department of Statistics, Harvard University }\\
\texttt{ shihaoyang@g.harvard.edu } \and \textsc{Yang Chen} \\
\textit{ Department of Statistics, Harvard University }\\
\texttt{ yangchen@fas.harvard.edu} \and \textsc{Neil Shephard} \\
\textit{Department of Economics and Department of Statistics, Harvard
University}\\
\texttt{shephard@fas.harvard.edu}  \and \textsc{Jun S. Liu\footnote{Corresponding author.}} \\
\textit{ Department of Statistics, Harvard University }\\
\texttt{ jliu@stat.harvard.edu } 
}
\date{\today}

\begin{document}

\maketitle

\vspace{-12pt}
\begin{abstract}
\noindent
We propose a weighting scheme for the proposals within Markov chain Monte Carlo algorithms and show how this can improve statistical efficiency at no extra computational cost. These methods are most powerful when combined with multi-proposal MCMC algorithms such as multiple-try Metropolis, which can efficiently exploit modern computer architectures with large numbers of cores.  The locally weighted Markov chain Monte Carlo method also improves upon a partial parallelization of the Metropolis-Hastings algorithm via Rao-Blackwellization. We derive the effective sample size of the output of our algorithm and show how to estimate this in practice.  Illustrations and examples of the method are given and the algorithm is compared in theory and applications with existing methods.

\vspace{6pt}
\end{abstract}

\newpage

\section{Introduction}\label{section:introduction}

Monte Carlo methods have become invaluable tools for solving demanding computational problems in a wide variety of scientific disciplines. In this paper, we propose weighting schemes for Markov chain Monte Carlo (MCMC) methods, where the main computational step often can be implemented using modern computer architectures with large numbers of cores. Since the weighting occurs within each iteration, we call this method locally weighted Markov chain Monte Carlo (LWMCMC). 

\vspace{12pt}
We will show that by allowing the points proposed in an MCMC algorithm, even those rejected, to take on weights, we can often improve statistical efficiency. The usual MCMC algorithms arise as special cases under specific weighting and proposal schemes in the framework we define. A measure of effective sample size ($ESS$) for this new class of algorithms is derived and shown to have natural connections to the existing measure of $ESS$ for MCMC \citep[p.~126]{kass1998,liu2001}. LWMCMC improves the parallel Metropolis-Hastings method of \citet{calderhead2014}. We show that our method can be interpreted as a Rao-Blackwellization of an extended version of his result.

\vspace{12pt}
To illustrate the idea, we first describe our weighting scheme for the Metropolis-Hastings (MH) algorithm \citep{metropolis1953, hastings1970}. To facilitate later discussion, our exposition of the algorithm differs sligthly from the standard description. Let the target density $\pi$ be defined on a sample space $S$.  The proposal kernel $K(dx_1;x)$ is a measure on $S$, with corresponding density $k(x_1;x)$. Set $x^{(1)}_0$ as the initial value and let $j=1$. To estimate $\mu_h = \int_S h(x)\pi(x)dx$, the MH algorithm iterates:
\textit{
\begin{enumerate}
\item Draw a proposal $x^{(j)}_1$ from $K(dx_1; x^{(j)}_0)$.
\item Calculate
      $$r(x^{(j)}_1;x^{(j)}_0) = \min\left\{1,\frac{\pi(x^{(j)}_1)k(x^{(j)}_0;x^{(j)}_1)}{\pi(x^{(j)}_0)k(x^{(j)}_1;x^{(j)}_0)} \right\}.$$
\item Set $y=x^{(j)}_1$ with probability $r(x^{(j)}_1;x^{(j)}_0)$ and $y=x^{(j)}_0$ with probability $1-r(x^{(j)}_1;x^{(j)}_0)$. \label{step:mh_draw}
\item Set $x^{(j+1)}_0=y$, set $j=j+1$ and go to step 1 until $j=n$.
\item Estimate $\mu_h$ with $\frac{1}{n}\sum_{j=1}^n h(x_0^{(j)}).$ \label{step:estimate}
\end{enumerate}}

In this paper, we propose giving both $x^{(j)}_0$ and $x^{(j)}_1$ weights $w(x^{(j)}_0)$ and $w(x^{(j)}_1)$ for each $j$ and substitute step \ref{step:estimate} with the new LWMCMC estimator
$$\hat{\mu}_h = \frac{1}{n}\sum_{j=1}^n \sum_{i=0}^1 w(x^{(j)}_i)h(x^{(j)}_i).$$
For instance, taking  $w(x^{(j)}_0) = 1-r(x^{(j)}_1;x^{(j)}_0)$ and $w(x^{(j)}_1) = r(x^{(j)}_1;x^{(j)}_0)$ results in an unbiased $\hat{\mu}_h$ which often has lower variance than the standard MH estimator. We will primarily be looking at two weighting schemes that give unbiased estimators, of which the aforementioned is the first version. Version 2 uses the weights
$$w(x^{(j)}_0) = \frac{\pi(x^{(j)}_0)k(x^{(j)}_1;x^{(j)}_0)}{\sum_{i=0}^1 \pi(x^{(j)}_i)k(x^{(j)}_{1-i};x^{(j)}_i)}, \quad w(x^{(j)}_1) = \frac{\pi(x^{(j)}_1)k(x^{(j)}_0;x^{(j)}_1)}{\sum_{i=0}^1 \pi(x^{(j)}_i)k(x^{(j)}_{1-i};x^{(j)}_i)}.$$

One way to systematically construct weighting schemes that result in unbiased estimators is by noting that step \ref{step:mh_draw} is a move from $x^{(j)}_0$ in a finite state Markov chain on $\{x^{(j)}_0, x^{(j)}_1\}$ defined by the transition matrix
\[ P =  \left( \begin{array}{cc}
1 - r(x^{(j)}_1;x^{(j)}_0) & r(x^{(j)}_1;x^{(j)}_0)  \\
r(x^{(j)}_0;x^{(j)}_1) & 1-r(x^{(j)}_0;x^{(j)}_1)  \\ \end{array} \right)\]

and that substituting step \ref{step:mh_draw}  with setting $y=x_1^{(j)}$ with probability $P^{\nu}_{1,2}$ and $y=x_0^{(j)}$ with probaility $P^{\nu}_{1,1}$ for any $\nu \geq 1$  leaves $\pi$ invariant. Here, $P^{\nu}_{i,j}$ is the $(i,j)$ entry of the $\nu\mbox{th}$ power of $P$. Hence, taking the weights $w(x^{(j)}_0) = P^{\nu}_{1,1}$ and $w(x^{(j)}_1) = P^{\nu}_{1,2}$ in $\hat{\mu}_h$ results in an unbiased estimator. In particular, version 2 of the weights arises from taking  ${\nu}\rightarrow \infty$ which corresponds to the stationary distribution of the Markov chain defined by $P$. It is easy to show that the version 2 weights satisfy $\bm w^{(j)} = \bm w^{(j)}P$, where $\bm w^{(j)} = \{w(x_0^{(j)}),w(x_1^{(j)})\}$. They also appear in the acceptance-rejection rule of \citet{barker1965}. 

\vspace{12pt}
Moreover, the usual MH algorithm itself can be viewed as producing locally weighted samples, in which the accepted point  gets weight 1 and the rejected gets weight  0 in step \ref{step:estimate}. The gain in choosing other weighting schemes stems partly from the reduction in variation of the weights. The weighting scheme used to produce $\hat{\mu}_h$ can be chosen independently of the probability vector used to propagate the chain. Section \ref{section:lwmcmc} proves the unbiasedness of $\hat{\mu}_h$ constructed the way outlined above for a generalized version of the MH algorithm.

\vspace{12pt}
In section \ref{section:ess} we show how to compute the effective sample size ($ESS$) for a given chain and weighting scheme. Applying this measure to $n=10,000$ samples obtained using the MH algorithm targeting a two-dimensional standard normal distribution using a normal proposal kernel with covariance $1.2^2 I_2$, in combination with the usual MH, $\nu = 1$ and $\nu \rightarrow \infty$ weights, gives $ESS = 1,189$ and $1,359$ and $1,337$ respectively. Since this is simply using the exact same samples weighted in three different ways, it shows we can trivially improve the estimation procedure by using locally weighted samples. The proposal covariance was tuned to give an acceptance rate of roughly 50\% as recommended by \citet{acceptance1997}.
 
\vspace{12pt}
The rest of this paper has four sections.  In section \ref{section:lwmcmc} we detail the main idea in a wider context. Section \ref{section:ess} gives a way of computing the algorithm's $ESS$. Section \ref{section:examples} provides illustrations of the method. Section \ref{section:conclusion} concludes, while the appendix contains the relevant proofs.

\section{Locally weighted Markov chain Monte Carlo}\label{section:lwmcmc}
\subsection{Main Idea}
We begin by giving a more general algorithm than the above example. It naturally extends multi-proposal MCMC algorithms such as the multiple-try Metropolis (MTM) of \citet{mtm2000}, and can beat such methods by not discarding potentially useful information available in the MCMC. By making multiple proposals within each iteration, MTM allows the use of transition kernels corresponding to large searching regions, hence mediating the ``conflict of interest" between the desired step size and desired acceptance rate that arises in MH. According to a sampling rule on the set of proposals \citep{mtm2000, rgm2001} MTM chooses one candidate point for an acceptance-rejection step. Typically, this is a ``good" point which is often accepted. In addition to performing such an acceptance-rejection step, we also advocate using the available weights to give estimators with smaller variance than the standard MCMC estimator.

\vspace{12pt}
As before, we wish to sample from the target distribution $\pi$, known up to a normalizing constant, defined on a sample space $S$. Let $K(dx_1,\dots,dx_M;x)$ denote a one-to-M kernel, defined as a probability measure on $S^{\otimes M}$ with density function $k(x_1,\dots,x_M;x)$. $K$ summarizes the proposal generation process, and allows for dependency and deterministic relations among the proposed points. Define a $(M+1)$-to-one kernel $T(dy;x_0,x_1,\dots,x_M)$, which is a probability measure on $S$. The restriction on $T$ is that it has to leave $\pi$ invariant, but can be taken to be any acceptance-rejection step or sampling rule that satisfies this. Examples are given in sections \ref{section:choice} and \ref{section:examples}.

\vspace{12pt}
\begin{algorithm}[Locally weighted MCMC]  \label{algo:mcis} Set $x_0^{(1)}$ to be the initial value and set $j=1$. Collect points $\{x_i^{(j)} : i=0,\dots, M; j=1,\dots n\}$ and weights $\{w(x_i^{(j)}) : i=0,\dots, M; j=1,\dots n\}$ to estimate $\mu_h$ according to the steps:
\begin{enumerate}
\item Draw proposals $\{x^{(j)}_1,\dots,x^{(j)}_M\}$ from $K(dx_1,\dots,dx_M;x^{(j)}_0)$. \label{step:mcis_draw}
\item Calculate and store the weights $w(x^{(j)}_i)$ according to a weighting scheme chosen by the user, e.g. (but not restricted to): \label{step:mcis_resample}
	\begin{itemize}
	\item Version 1: $$w(x^{(j)}_i) = \frac{1}{M}\min \left\{1, \frac{\pi(x^{(j)}_i)k(\bm x^{(j)}_{-i};x^{(j)}_i)}{\pi(x^{(j)}_0)k(\bm x^{(j)}_{-0};x^{(j)}_0)} \right\},\quad \mbox{for $i\geq 1$, } \quad w(x^{(j)}_0) = 1- \sum_{i=1}^M w(x^{(j)}_i).$$
	\item Version 2: $$w(x^{(j)}_i) = \frac{\pi(x^{(j)}_i)k(\bm x^{(j)}_{-i};x^{(j)}_i)}{\sum_{i=0}^M \pi(x^{(j)}_i)k(\bm x^{(j)}_{-i};x^{(j)}_i)}, \quad \mbox{for} \quad i = 0,\dots,M, \quad\quad\quad\quad\quad\quad\quad\quad\quad\quad\quad\quad\quad\quad\quad$$
where $\bm x^{(j)}_{-i} = (x^{(j)}_{0},\dots,x^{(j)}_{i-1},x^{(j)}_{i+1},\dots,x^{(j)}_{M})$.
	\end{itemize}
\item Draw $y$ from $T(dy;x^{(j)}_0,x^{(j)}_1,\dots,x^{(j)}_M)$. \label{step:mcis_propogation}
\item Set $x^{(j+1)}_0=y$ and $j=j+1$, and go to step \ref{step:mcis_draw} until $j=n$.
\item Estimate $\mu_h$ with $\hat{\mu}_h = \frac{1}{n}\sum_{j=1}^n \sum_{i=0}^M h(x_i^{(j)})w(x_i^{(j)}).$
\end{enumerate}
\end{algorithm}

The weights are normalised within each iteration, such that $\sum_{i=0}^M w(x_i^{(j)}) = 1$ for all $j$. Version 1 corresponds to the multi-proposal version of the $\nu = 1$ weights from section \ref{section:introduction}. Likewise, version 2 uses the analogous $\nu \rightarrow \infty$ weights. Other weighting schemes, e.g. for $\nu \geq 2$, are potentially also useful, but typically require more computation.

\subsection{Choice of the propagation kernel $T$} \label{section:choice}
$T$ can be taken as any acceptance-rejection step or resampling rule that leaves $\pi$ invariant. For instance,  $T$ can be taken independent of the set of proposals $\{x_1,\dots,x_M\}$ generated by $K$. The class of such $T$'s essentially contains all MCMC methods. Examples include sampling $y$ according to a standard MH or Gibbs step from $x$, for which $T(dy;x_0,x_1,\dots,x_M) = T(dy;x_0)$.

\vspace{12pt}
Perhaps more useful is to allow $T$ to use information about the set $\{x_1,\dots,x_M\}$. E.g. letting $T(dy;x_0,x_1,\dots,x_M) = \sum_{i=0}^M w(x_i) \delta_{x_i}(dy)$, where $\delta_x$ denotes the Dirac delta function, we encompass Calderhead's algorithm (see the appendix). $T$ allows us to store weights according to one weighting scheme, but propagate using another. Moreover, the algorithm also has flexibility to use MTM rules on $\{x_1,\dots,x_M\}$ to choose a ``good" point to move to. An explicit example is given in \ref{section:hmc}.

\subsection{Properties} \label{section:properties}
\begin{theorem} \label{theorem:proper}
The estimators produced by versions 1 and 2 of Algorithm \ref{algo:mcis} are unbiased for $\mu_h$ for any measurable $h$.
\end{theorem}

The proof is given in the appendix. It relies on Algorithm \ref{algo:mcis_embedded_mcmc}, which encodes the weights arising in Algorithm \ref{algo:mcis} with an empirical distribution. Algorithm \ref{algo:mcis_embedded_mcmc} is itself a generalization of Calderhead's algorithm, introducing the flexibility of $T$ in the propagation step. The empirical distribution introduces Monte Carlo error, which is the subject of Theorem \ref{theorem:rb}:

\begin{theorem} \label{theorem:rb}
Given the same weighting scheme, Algorithm \ref{algo:mcis} is a Rao-Blackwellization of Algorithm \ref{algo:mcis_embedded_mcmc}.
\end{theorem}

In the special case of $T$ being Calderhead's propagation rule, Algorithm \ref{algo:mcis} is a Rao-Blackwellization of Calderhead's algorithm. Again, the proof is given in the appendix.

\section{Effective Sample Size of LWMCMC}\label{section:ess}
To evaluate LWMCMC, we derive a measure of effective sample size. If $\hat{\mu}$ is the standard estimator of the mean based on the samples and their weights, and $\sigma^2$ is the variance of $\pi$, $ESS$ is defined $ESS = \sigma^2/ \Var(\hat{\mu}).$ Recall that the output of Algorithm \ref{algo:mcis} is $n(M+1)$ weighted samples, producing the mean estimate
$$\hat{\mu}  = \frac{1}{n} \sum_{j=1}^n \bar{x}^{(j)} , \quad \mbox{where} \quad \bar{x}^{(j)} = \sum_{i=0}^M w(x^{(j)}_i)x^{(j)}_i.$$

\begin{proposition} \label{theorem:ess}
The $ESS$ for samples and weights on the form produced by Algorithm \ref{algo:mcis} can be written as
$$ESS = \frac{n}{\frac{\Var(\bar{x})}{\sigma^2}\left(1+2\sum_{k}\gamma_k \right)},$$
where $\gamma_k$ is the lag-$k$ autocorrelation function of $\left \{\bar{x}^{(j)} \right\}^n_{j=1}$ and $\Var(\bar{x}^{(j)}) = \Var(\bar{x})$ for all $j$ by stationarity.
\end{proposition}

The proof is given in the appendix. In the case of the usual MH and its multi-proposal extensions, one always takes $w(x_0^{(j)}) = 1$ and $w(x_i^{(j)}) = 0$ for $i>1$. Then $\bar{x}^{(j)} = x_0^{(j)}$ and $x_0^{(j)} \sim \pi$ for all $j$ assuming the chain has converged, so that $\Var(\bar{x}^{(j)})= \Var(\bar{x}) = \sigma^2$ for all $j$. Moreover, $\gamma_k = \rho_k$ where $\rho_k$ is the lag-$k$ autocorrelation of the Markov chain $\{x_0^{(j)}\}_{j=1}^n$. So the above expression reduces to the standard measure of $ESS$ in the case of MCMC, namely $ESS = n/(1+2\sum_k \rho_k)$.

\vspace{12pt}
To estimate $ESS$ for LWMCMC, substitute $1+2\sum_{k}\gamma_k$ with an estimate of the spectral density of $\left \{\bar{x}^{(j)} \right\}^n_{j=1}$ at frequency 0 \citep[see e.g.][]{andrews1991, muller2014}, and $\sigma^2$ and $\Var(\bar{x})$ with their respective moment estimators.

\vspace{12pt}
Given exactly the same propagating chain, LWMCMC beats standard MCMC in terms of $ESS$ whenever $\Var(\bar{x})(1+2\sum_{k}\gamma_k ) < \sigma^2(1+2\sum_{k}\rho_k)$, which is an easy condition to check. A good kernel $K$ will typically make $\Var(\bar{x})$ small, while $T$ is important in making $1+2\sum_{k}\gamma_k$ small.

\section{Numerical example}\label{section:examples}
We apply Algorithm \ref{algo:mcis}  to the two-dimensional conditional density 
$$\pi(z,\theta | y) \propto \exp\left\{ -\frac{(y-\theta z)^2}{2\sigma^2} - \frac{(z-\theta)^2}{2}\right\},$$ 

which arises under a flat prior on $\theta$ in the indirect observation model
$$y=\theta z + \epsilon, \quad z | \theta \sim N(\theta,1), \quad \epsilon \sim N(0, \sigma^2), \quad \epsilon \independent (z,\theta).$$
\citep[See][for further details]{chen2001}. Our interest lies in estimating the mean of $\theta$. When $\sigma$ is small most of the density degenerates to lie around the curve $y=\theta z$, making sampling particularly hard. Figure 1 shows the contours of the density when $\sigma = 0.1$ and $y$ is observed to be 1. 

\subsection{Locally weighted Metropolis-Hastings}
To illustrate the algorithm in a simple case, we first implement the very blunt random walk MH and its exact LWMCMC counterpart given in section \ref{section:introduction}. $K$ is  taken to sample  $(z_1^{(j)},\theta_1^{(j)}) = x_1^{(j)} \sim N(x_0^{(j)},\lambda^2 I_2)$ and $T$ performs the usual MH rejection step. Figures 2 and 3 show the samples obtained using MH and LWMCMC MH respectively with $n=10,000$. The step size of the RW MH was tuned to maximize $ESS$, corresponding to $\lambda=0.45$.

\vspace{12pt}
In Figure 3, the black dots represent the weighted samples produced by $K$, even those rejected in the propagation step. The red dots represent the samples produced by $T$, i.e. the actual propagating samples. Many of the black dots have close to zero weight, and hence do not bias our estimate of the mean of $\theta$. $ESS$ was computed to be 195 for MH,  220 for LWMCMC with $\nu=1$ and 216 for $\nu \rightarrow \infty$ when using exactly the same propagating chain in the three cases, showing only small increased efficiency of LWMCMC in this setting. It is clear that neither of these algorithms are well suited to this problem.

\begin{figure}[h]
\centering
\begin{minipage}{0.33\textwidth}
\centering
\includegraphics[scale=0.3]{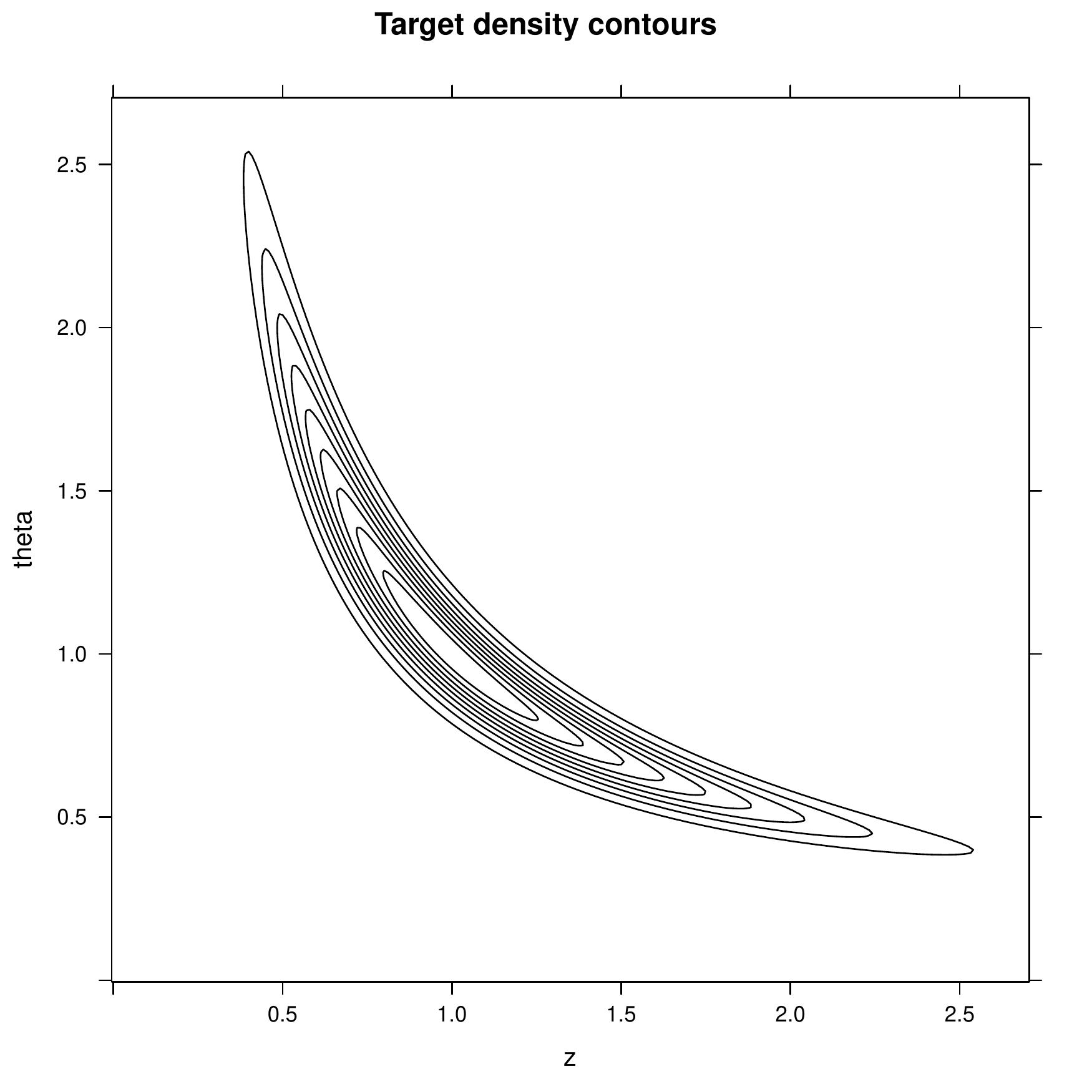}
\caption{Contours of $\pi$.}
\end{minipage}\hfill
\begin{minipage}{0.33\textwidth}\label{figure:mclis}
\centering
\includegraphics[scale=0.3]{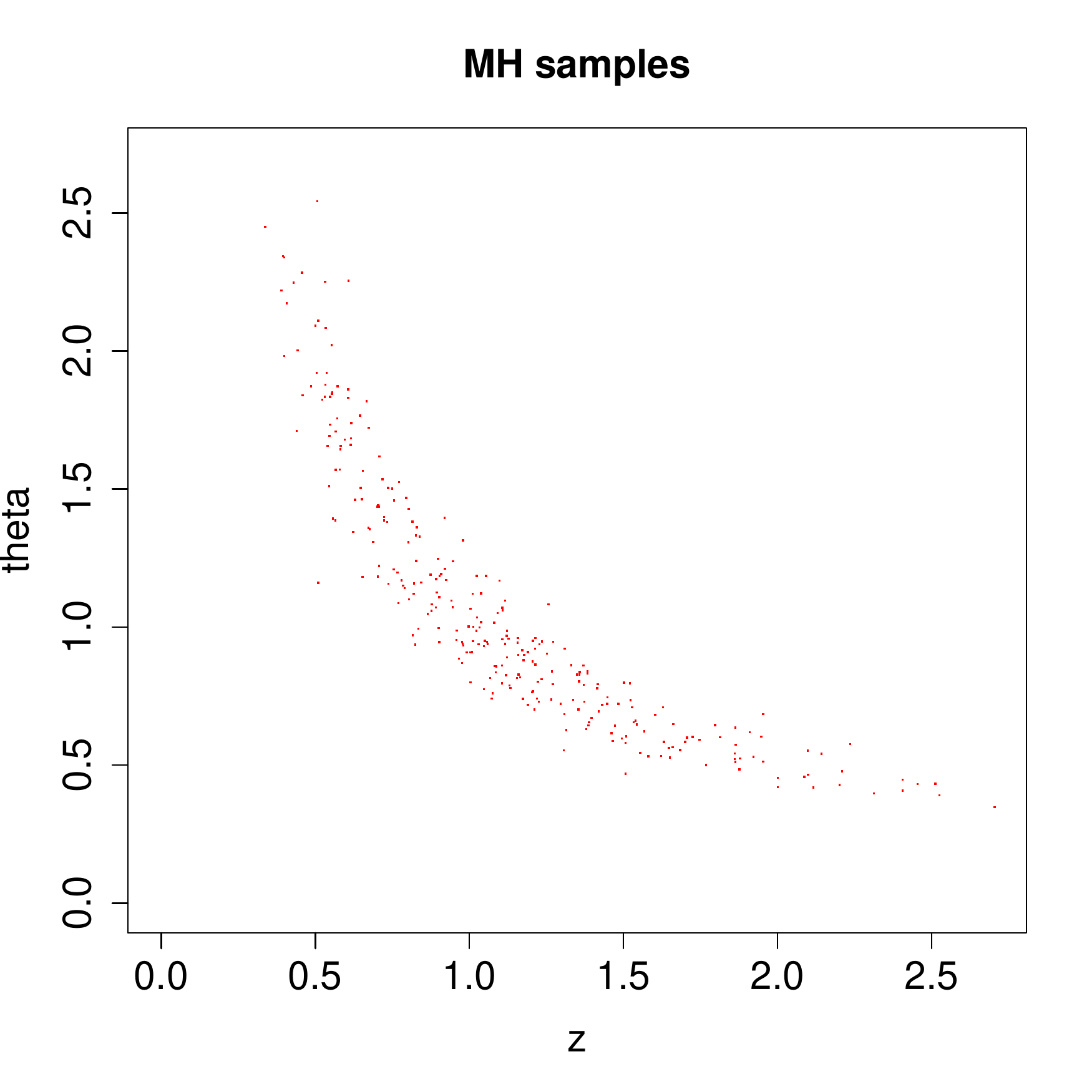}
\caption{RW MH.}
\end{minipage}
\begin{minipage}{0.33\textwidth}\label{figure:hmc}
\centering
\includegraphics[scale=0.3]{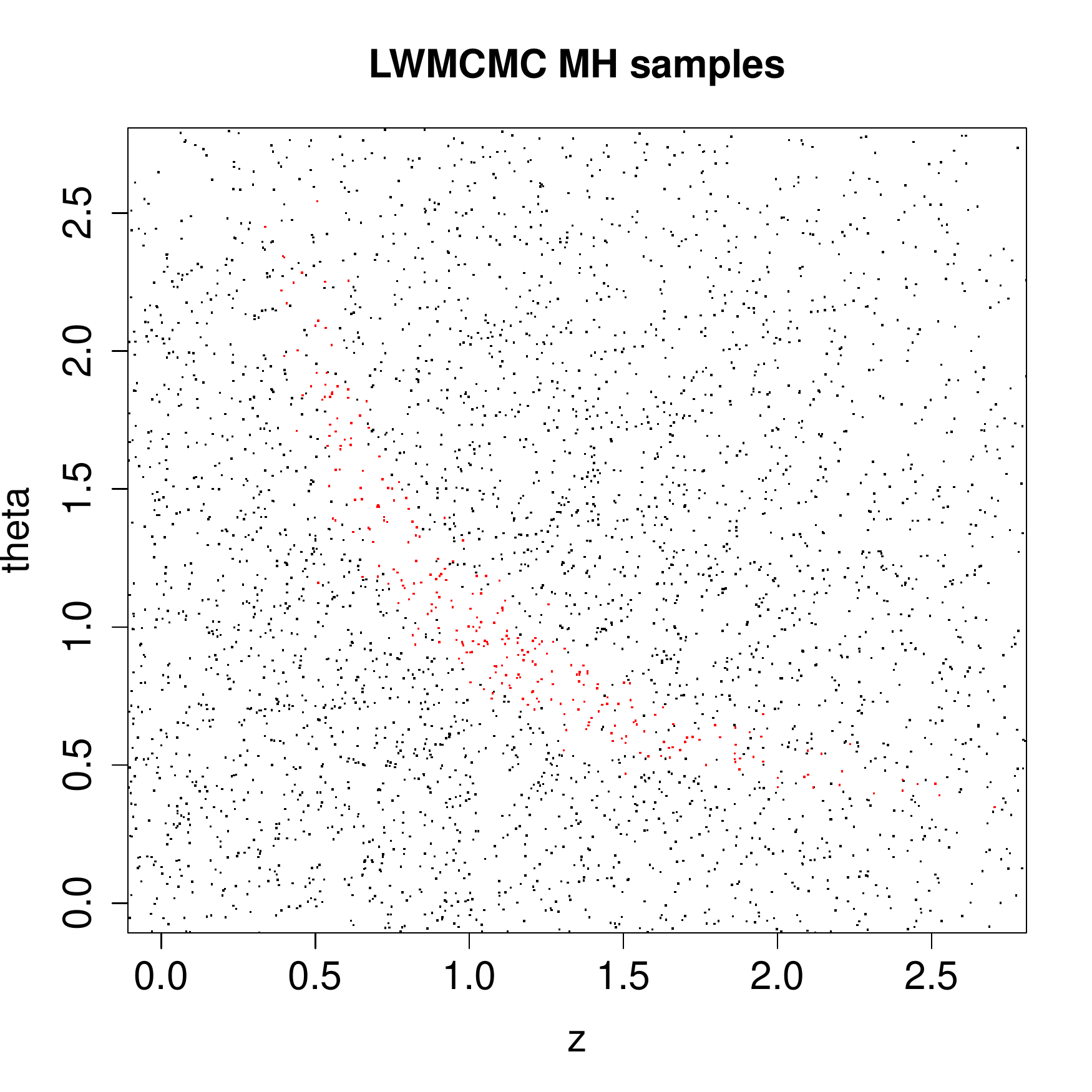}
\caption{LWMCMC RW MH.}
\end{minipage}
\end{figure}

\subsection{Locally weighted Hamiltonian Monte Carlo} \label{section:hmc}
The real benefit of LWMCMC arises when multiple points are proposed within each iteration. We illustrate one such algorithm here, where the leapfrog integration path that arises in Hamiltonian Monte Carlo (HMC) \citep{duane1987} is taken as the proposals in Algorithm \ref{algo:mcis}. In HMC the state $x = (z,\theta)^T$ is augmented with an auxiliary momentum variable $p$. The Hamiltonian is defined
$$H(x,p) = - \log \pi(x) + \frac{1}{2}p^TW^{-1}p,$$
where W is symmetric, positive-definite and is chosen by the user, and $p$ has the dimension of $S$. As the system evolves in time it obeys Hamilton's equations
$$ \frac{\partial x}{\partial t} = \frac{\partial H(x,p)}{\partial p} = W^{-1}p, \quad \frac{\partial p}{\partial t}= - \frac{\partial H(x,p)}{\partial x} = \{\nabla \log \pi(x)\}^T,$$
by conservation of energy. These equations need to be solved using numerical methods. The literature typically favors leapfrog integration, since it is reversible and outperforms Euler discretization. The leapfrog algorithm we use iterates
\begin{align*}
& x_{t + \delta/2} = x_t + \frac{\delta}{2} W^{-1}p_{t}, \\ 
& p_{t+\delta} = p_t - \delta\{\nabla \log \pi(x_t)\}^T, \\
& x_{t + \delta} = x_t + \frac{\delta}{2} W^{-1}p_{t+\delta}.
\end{align*}
Since $x_t$ is typically much cheaper to evaluate than $p_t$, this is of the same order of computational complexity as the standard leapfrog method.

\vspace{12pt}
\citet{duane1987} realized the following: If we can simulate from the augmented distribution $\pi^*(x,p) \propto \exp\{-H(x,p)\}$ then, marginally, $x \sim \pi$ and $\phi(p) \propto \exp(-0.5p^TW^{-1}p)$ so that $p$ is normally distributed. If $x_0^{(j)}$ is the state of the chain at iteration $j$, the HMC algorithm performs the steps
{\it
\begin{enumerate}
\item Draw a momentum vector $p_0^{(j)} \sim \phi$. \label{step:hmc_mom}
\item Starting from $(x^{(j)}_0,p^{(j)}_0)$, run the leapfrog integrator $M$ steps using a time increment of $\delta$ to obtain the proposal $(x_M^{(j)},p_M^{(j)})$.
\item Set $x_0^{(j+1)}= x_M^{(j)}$ with probability
$$r = \min[1, \exp\{-H(x_M^{(j)},p_M^{(j)}) + H(x_0^{(j)},p_0^{(j)})\}],$$
set $j=j+1$, and go to step \ref{step:hmc_mom} until $j = n$.
\end{enumerate}
}

HMC has been shown to be particularly useful for densities of the kind we are sampling from here, as Hamilton's equations prevent the proposals from escaping the energy well induced by the density. Note that when $M=1$, HMC reduces to the Metropolis-adjusted Langevin algorithm \citep{mala1996}. For an HMC algorithm performing $M>1$ leapfrog steps, \citet{neal1994} introduced the idea of sampling $l$ uniformly from the set $\{0,\dots,M\}$ and running the leapfrog integrator $l$ steps backwards and $M-l$ steps forwards, which was further developed in \citet{rgm2001}. This introduces symmetry among the points on the leapfrog path, which allows us to construct a $K$ that can be used in LWMCMC.

\vspace{12pt}
Specifically, sample $l$ as above and set $x^{(j)}_l = x$. Note here that the random index $i=l$ takes on the same meaning as the index $i=0$ did earlier. In parallel, run leapfrog integration backward in time for $l$ steps, generating  $\{x^{(j)}_0,\dots,x^{(j)}_{l-1}\}$, and forward for $M-l$ steps, generating $\{x^{(j)}_{l+1},\dots,x^{(j)}_M\}$. The set of associated momentum vectors is $\{p^{(j)}_0,\dots,p^{(j)}_M\}$. Let $K$ denote the measure that generates this proposal process. A challenge with this $K$ is that we are not able to fully exploit the trivial parallelization potential of the algorithm, since the integrator is sequential in nature.

\vspace{12pt}
By the symmetric sampling of $l$, the $\nu \rightarrow \infty$ weights of the proposals reduce to
$$w(x^{(j)}_i) = \frac{\exp\{-H(x^{(j)}_i,p^{(j)}_i)\}}{\sum_{i=0}^M \exp\{-H(x^{(j)}_i,p^{(j)}_i)\}}, \quad \mbox{for $i = 0,\dots, M$}.$$
For the same reason, the $\nu = 1$ weights are
$$w(x^{(j)}_i)  = \frac{1}{M}\min\left\{1, \frac{\exp\{-H(x^{(j)}_i,p^{(j)}_i)\}}{\exp\{-H(x^{(j)}_0,p^{(j)}_0)\}}\right\}, \quad \mbox{for $i \geq 1$, and} \quad w(x^{(j)}_0) = 1-\sum_{i=1}^M  w(x^{(j)}_i).$$

\vspace{12pt}
Moreover, let $a = 0$ if $l > M-l$ and $a=M$ otherwise. Take $T$ be the measure that performs the usual HMC Metropolis step comparing the initial point $x_l^{(j)}$ with $y= x_a^{(j)}$. That is, accept $x_a^{(j)}$ with probability
$$r = \min[1, \exp\{-H(x_a^{(j)},p_a^{(j)}) + H(x_l^{(j)},p^{(j)}_l)\}].$$

\vspace{12pt}
Table 1 summarizes the results of conventional HMC against LWMCMC HMC and how $ESS$ scales with $M$, with $n=1,000$, $\delta = 0.05$ and $W=I_2$. We also compare LWMCMC to Calderhead's algorithm for different values of $M$ and the resampling parameter $N$. See algorithm \ref{algo:mcis_embedded_mcmc} and remarks \ref{remark:comparison} and  \ref{remark:empirical} of the appendix for further details and comparisons between the algorithms. We apply the same measure of effective sample size, noting that Calderhead's algorithm induces a weighting scheme as mentioned in remark \ref{remark:empirical}. The improvement in $ESS$ of LWMCMC against Calderhead stems both from choosing a more useful $T$ and from the Rao-Blackwellization discussed in section \ref{section:properties}. The improvements in LWMCMC and Calderhead's method as $M$ increases are due to the decreasing variance of $\bar{x}$. Note also that we are observing super-efficiency in the HMC sampling scheme, where negative correlation among the samples lead to $ESS$ greater than the number of samples drawn.

\vspace{12pt}
Figures 4 and 5 display the HMC and LWMCMC HMC samples obtained when $M=60$. In Figure 5, the black dots represent all the positions visited in the leapfrog integration, and hence are the weighted samples produced by $K$. The red dots represent the samples produced by $T$, i.e. the propagating samples.
\begin{figure}[h]
\centering
\begin{minipage}{0.5\textwidth}
\centering
\includegraphics[scale=0.32]{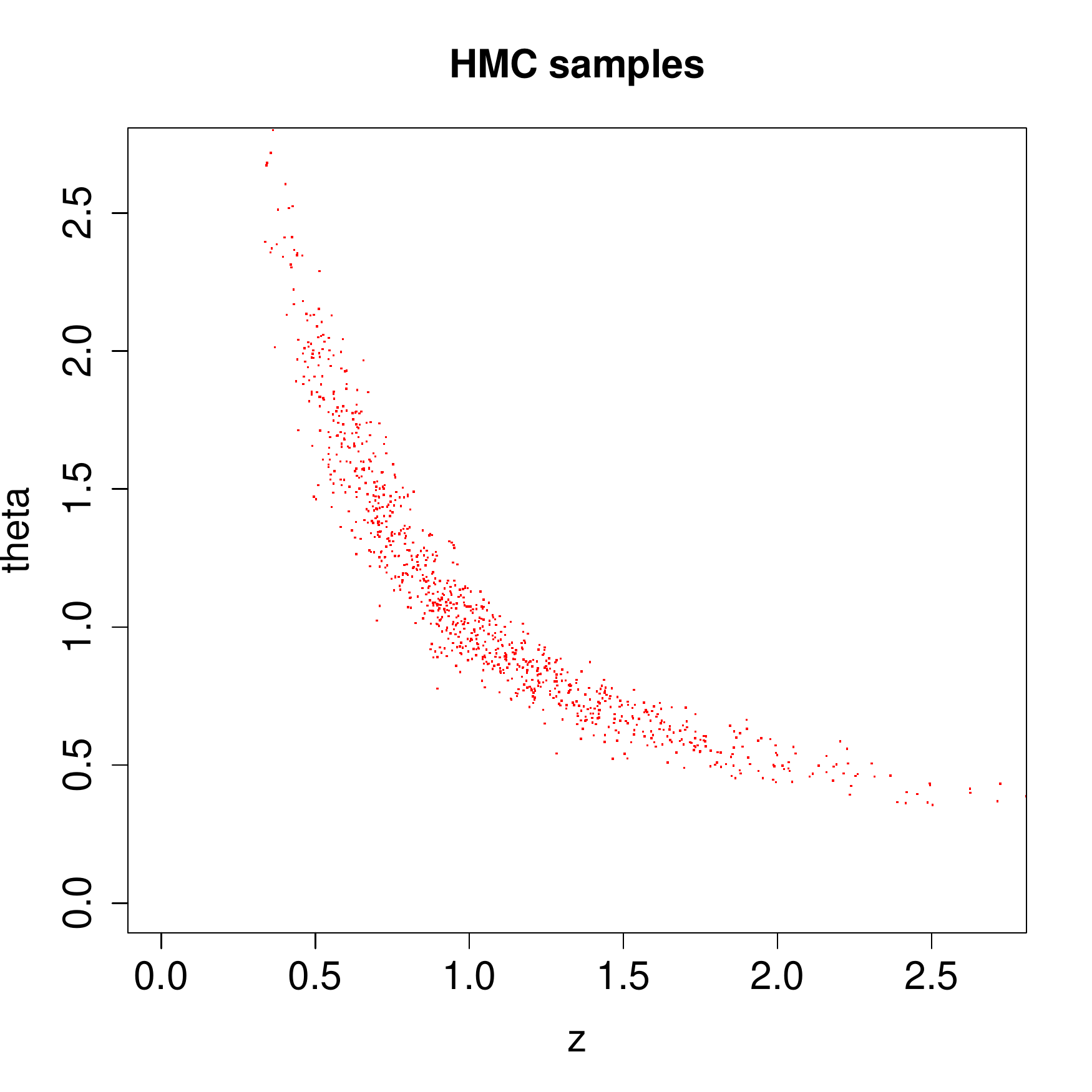}
\caption{HMC.}
\end{minipage}\hfill
\begin{minipage}{0.5\textwidth}\label{figure:mclis}
\centering
\includegraphics[scale=0.32]{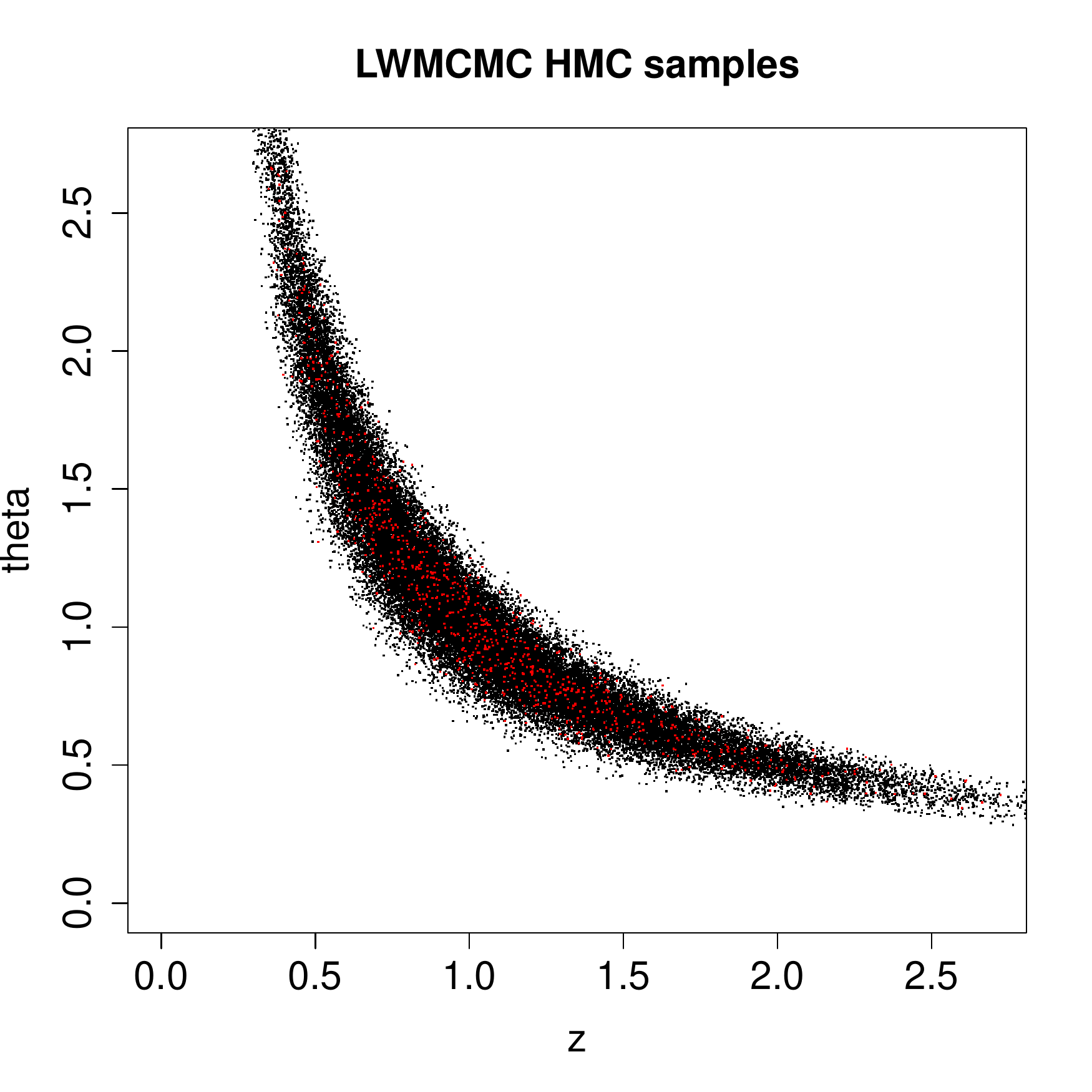}
\caption{LWMCMC HMC.}
\end{minipage}
\end{figure}

\begin{table}[h]
\begin{center}
\begin{tabular}{r | r | r  | r  || r | r | r | r | r}
\multicolumn{9}{c}{Effective Sample Size ($ESS$)} \\ \hline \hline
\multicolumn{1}{c|}{} & \multicolumn{2}{c|}{LWMCMC HMC}& \multicolumn{1}{c||}{HMC} & \multicolumn{5}{c}{Calderhead HMC, $\nu \rightarrow \infty$} \\ \hline 
  $M$ & $\nu = 1$  & $\nu \rightarrow \infty$ &  & $N=1$ & 10 & 50 & 200 & 1000 \\ \hline
     5 &          39 &         39 &   26  & 7.6 & 8.0 & 8.0 & 8.0 & 8.0 \\
     30 &      624 &      625  &  678  & 158 & 185 & 202 & 203  & 203\\
     60 &   2,468 &  2,483   &  1,137  & 440 & 1,333 & 1,544 & 1,525 & 1,512 \\
     90 &  4,773  &    4,792  & 938  & 705 & 2,638 & 3,103 & 3,275  & 3,282 \\
     240 & 9,613 &    9,681  & 1,061  & 916 & 3,312 & 5,207 & 6,235 & 6,333 \\
\end{tabular}
\caption{$ESS$ for HMC, LWMCMC HMC, and Calderhead at different values of $N$.}
\end{center}
\end{table}

\section{Conclusions} \label{section:conclusion}
In this paper we have proposed the locally weighted Markov chain Monte Carlo algorithm (LWMCMC), which dominates its parallel MCMC counterpart and typically improves upon standard MCMC. We show how to compute the effective sample size of the LWMCMC output and illustrate its performance on a toy example. The LWMCMC algorithm is well suited to modern computer architectures with massive numbers of cores, which can dramatically increase computational efficiency.

\vspace{12pt}

\begin{appendix}
\section{Proofs and supplementary material}
\subsection{Proof of Theorem \ref{theorem:proper}}
{\bf Theorem \ref{theorem:proper}.} \textit{The estimators produced by Algorithm \ref{algo:mcis} are unbiased for $\mu_h$.}
\vspace{12pt}

Before proving the theorem, we prove the following lemma.

\begin{lemma} \label{lemma:mcis_embedded_mcmc}
Samples obtained according to Algorithm \ref{algo:mcis_embedded_mcmc} given below are draws from $\pi$.
\end{lemma}

\begin{algorithm} \label{algo:mcis_embedded_mcmc} Set $x_0^{(1)}$ to be the initial value and set $j=1$. Collect points $\{y_i^{(j)} : i=1,\dots, N; j=1,\dots n\}$ according to the steps:
\begin{enumerate}
\item Draw proposals $\{x^{(j)}_1,\dots,x^{(j)}_M\}$ from $K(dx_1,\dots,dx_M;x^{(j)}_0)$.
\item Sample $N$ points $\{y^{(j)}_1,\dots,y^{(j)}_{N}\}$ (with replacement) from $\{x^{(j)}_0,\dots,x^{(j)}_M\}$ with probabilities $w(x^{(j)}_i)$. \label{step:mcis_embedded_mcmc_resample}
\item Draw $y$ from $T(dy; x^{(j)}_0,x^{(j)}_1,\dots,x^{(j)}_M)$. \label{step:mcis_embedded_mcmc_propogation}
\item Set $x^{(j+1)}_0=y$ and $j=j+1$, and go to step 1 until $j=n$.
\item Estimate $\mu_h$ with $\hat{\mu}_h = \frac{1}{nN}\sum_{j=1}^n \sum_{i=1}^N h(y_i^{(j)}).$\end{enumerate}
\end{algorithm}

\begin{proof}[Proof of Lemma \ref{lemma:mcis_embedded_mcmc}]
$\{x^{(j)}_0 : j=1,\dots, n\}$ is a standard MCMC sampler by construction of the transition kernel $T$. Since Calderhead's algorithm is valid for versions 1 and 2 of the weighting scheme, we know step \ref{step:mcis_embedded_mcmc_resample} draws samples from $\pi$ given that $x^{(j)}_0$ follows $\pi$. Combining these two arguments we establish that the samples $\{y_i^{(j)} : i=1,2,\ldots, N; j=1,2,\ldots n\}$ all have marginal distribution $\pi$.
\end{proof}

\begin{proof}[Proof of Theorem \ref{theorem:proper}]
Let $y^{(j)}_i$ be a point drawn according to Algorithm \ref{algo:mcis_embedded_mcmc}. Let $\bm x^{(j)} = \{x^{(j)}_0,\dots,x^{(j)}_M\}$, so that
\vspace{-12pt}
\begin{align*}
\mu_h &= \E \left \{h(y^{(j)}_i)\right\} = \E\left [\E \left \{h(y^{(j)}_i) \Big| \bm x^{(j)} \right\}\right] = \E \left \{\sum^M_{i=0} w(x^{(j)}_i)h(x^{(j)}_i)\right\}
\end{align*}
\end{proof}

\subsection{Proof of Theorem \ref{theorem:rb} and remarks}
{\bf Theorem \ref{theorem:rb}.} \textit{Given the same weighting scheme, Algorithm \ref{algo:mcis} is a Rao-Blackwellization of Algorithm \ref{algo:mcis_embedded_mcmc}.}
\begin{proof}
Let $\bm x^{(j)} = \{x^{(j)}_0,\dots,x^{(j)}_M\}$ and  $\bm y^{(j)} = \{y^{(j)}_1,\dots,y^{(j)}_N\}$.

\begin{align*}
\Var \left\{ \frac{1}{nN} \sum_{j=1}^n \sum_{i=1}^N h(y^{(j)}_i) \right\} &= \frac{1}{n^2N^2} \sum_{j=1}^n \Var\left\{ \sum_{i=1}^N h(y_i^{(j)})\right\} + \frac{2}{n^2N^2}\sum_{j<k}^n \Cov \left\{ \sum_{i=1}^N h(y_i^{(j)}),\sum_{i=1}^N h(y_i^{(k)})\right\}.
\end{align*}

\vspace{12pt}
By the law of total variance, for the first of these terms we have
\begin{align*}
\frac{1}{n^2N^2}\sum_{j=1}^n \Var\left\{ \sum_{i=1}^N h(y_i^{(j)})\right\} &\geq \frac{1}{n^2N^2}\sum_{j=1}^n \Var\left[ \E \left\{ \sum_{i=0}^N h(y_i^{(j)}) | \bm x^{(j)} \right\} \right] \\
& = \frac{1}{n^2} \sum_{j=1}^n \Var \left\{\sum_{i=0}^M w(x_i^{(j)})h(x_i^{(j)}) \right\}.
\end{align*}
\vspace{12pt}

For the second, we will show that
\begin{align*}
\frac{2}{n^2N^2}\sum_{j<k}^n \Cov \left\{ \sum_{i=1}^N h(y_i^{(j)}),\sum_{i=1}^N h(y_i^{(k)})\right\}
&= \frac{2}{n^2}\sum_{j<k}^n \Cov \left\{ \sum_{i=0}^M w(x_i^{(j)})h(x_i^{(j)}),\sum_{i=0}^M w(x_i^{(k)})h(x_i^{(k)})\right\}.
\end{align*}

Assume $j<k$ without loss of generality. By the law of total covariance we then have
\begin{align*}
&\Cov \left\{\frac{1}{N}\sum_{i=1}^N h(y_i^{(j)}),\frac{1}{N}\sum_{i=1}^N h(y_i^{(k)})\right\}\\
=& \; \Cov \left[\frac{1}{N}\E \left\{\sum_{i=1}^N h(y_i^{(j)}) \Big| \bm x^{(j)}, \bm x^{(k)}\right\},\frac{1}{N}\E \left\{\sum_{i=1}^N h(y_i^{(k)}) \Big| \bm x^{(j)}, \bm x^{(k)}\right\} \right]\\
& \quad + \E \left[\Cov \left\{\frac{1}{N}\sum_{i=1}^N h(y_i^{(j)}),\frac{1}{N}\sum_{i=1}^N h(y_i^{(k)}) \Big| \bm x^{(j)}, \bm x^{(k)}\right\}\right]\\
= & \;\Cov \left\{\sum_{i=0}^M w(x_i^{(j)})h(x_i^{(j)}),\sum_{i=0}^M w(x_i^{(k)})h(x_i^{(k)})\right\},
\end{align*}

where the last equality follows from the conditional independence structure of $\bm y^{(j)}$ from all the other samples and proposals given $\bm x^{(j)}$. Summarizing these results, we can conclude that
\begin{align*}
\Var \left\{ \frac{1}{nN} \sum_{j=1}^n \sum_{i=1}^N h(y^{(j)}_i) \right\} &\geq \Var \left\{ \frac{1}{n} \sum_{j=1}^n \sum_{i=0}^M w(x_i^{(j)})h(x^{(j)}_i)\right\}.
\end{align*}
\end{proof}

\begin{remark} \label{remark:comparison}
Hence, with the same amount of computational time, but less memory and no resampling step, we are able to do better than Algorithm \ref{algo:mcis_embedded_mcmc}. The proof of Theorem \ref{theorem:rb} illuminates where the reduction in variance occurs. This indicates that the method is still sensitive to the properties of the Markov chain used to propagate the sample space. In particular, our method will have the same degree of stickiness as Algorithm \ref{algo:mcis_embedded_mcmc}.
\end{remark}
\vspace{6pt}

\begin{remark} \label{remark:empirical}
Algorithm \ref{algo:mcis_embedded_mcmc} also induces a weighting scheme. Specifically, if $N_i^{(j)}$ is the number of times the proposal $x_i^{(j)}$ is resampled out of the $N$ resampled points, the weights are $N_i^{(j)}/N$. Note that $(N_0^{(j)}, \dots, N_M^{(j)}) \sim \mbox{Multinomial}[N, \{w(x_0^{(j)}, \dots, w(x_M^{(j)})\}]$. Thus, both Calderhead's algorithm and \ref{algo:mcis_embedded_mcmc} attempt to encode the information about this multinomial distribution using an empirical distribution. This creates a loss of information, as proved in Theorem $\ref{theorem:rb}$. The Dvoretszky-Kiefer-Wolfowitz inequality \citep{dkw1956} provides probability bounds for how close the empirical CDF of the resampled points is to the actual CDF as a function of $N$. This can give some indication of how large $N$ would have to be chosen for Algorithm \ref{algo:mcis_embedded_mcmc} and Calderhead's algorithm to approximate LWMCMC well. As $N \rightarrow \infty$, Algorithm \ref{algo:mcis_embedded_mcmc} converges to LWMCMC.
\end{remark}

\subsection{Derivation of $ESS$ for LWMCMC}

{\bf Proposition \ref{theorem:ess}} \textit{The $ESS$ for samples and weights on the form produced by Algorithm \ref{algo:mcis} can be written as
$$ESS = \frac{n}{\frac{\Var(\bar{x})}{\sigma^2}\left(1+2\sum_{k}\gamma_k \right)},$$
where $\gamma_k$ is the lag-$k$ autocorrelation function of $\left \{\bar{x}^{(j)} \right\}^n_{j=1}$ and $\Var(\bar{x}) = \Var(\bar{x}^{(j)})$ for all $j$ by stationarity.}

\begin{proof}
\begin{align*}
\frac{\sigma^2}{ESS} & = \Var\left(\frac{1}{n}\sum_{j=1}^n \bar{x}^{(j)}\right) = \frac{1}{n^2}\sum_{j=1}^n \Var(\bar{x}^{(j)}) +\frac{2}{n^2}\sum_{j<k}^n\Cov(\bar{x}^{(j)},\bar{x}^{(k)})\\
& = \frac{1}{n} \Var(\bar{x}) +\frac{2}{n}\sum_{k=1}^{n-1}\left(1-\frac{k}{n}\right)\gamma_k \Var(\bar{x}) \\
& = \frac{\Var(\bar{x})}{n}\left\{ 1 + 2\sum_{k=1}^{n-1}\left(1-\frac{k}{n}\right)\gamma_k \right\},
\end{align*}
where the second inequality follows from stationarity. Recall that by the Ces\`{a}ro summability theorem
$$ \lim_{n\rightarrow \infty} \sum_{k=1}^{n-1}\left(1-\frac{k}{n}\right)\gamma_k = \sum_k\gamma_k .$$

For sufficiently large $n$, we therefore substitute the right hand side of this equality into the expressions derived above. Rearranging the terms will give the desired result.
\end{proof}
\end{appendix}

\bibliography{MCIS_ref}
\bibliographystyle{apalike}

\end{document}